\documentclass[12pt]{article}
\usepackage{mathtools}
\usepackage{amssymb,amsthm}
\usepackage[mathscr]{euscript}
\usepackage{cite}
\usepackage[colorlinks=true,linkcolor=blue,citecolor=magenta]{hyperref}

\usepackage[utf8]{inputenc}
\usepackage[bitstream-charter]{mathdesign}
\usepackage[T1]{fontenc}

\usepackage{microtype}
\usepackage{enumerate}
\usepackage[letterpaper,hmargin=3cm,vmargin=2.5cm]{geometry}
\usepackage{setspace}
\makeatletter
\renewcommand{\section}{\@startsection%
{section}%
{1}%
{0em}%
{1.7em}%
{1.2em}%
{\normalfont\large\centering\bfseries}}
\renewcommand{\@seccntformat}[1]%
{\csname the#1\endcsname.\hspace{0.5em}}
\makeatother

\numberwithin{equation}{section}

\newtheorem{theorem}{Theorem}[section]
\newtheorem{proposition}[theorem]{Proposition}
\newtheorem{lemma}[theorem]{Lemma}
\newtheorem{corollary}[theorem]{Corollary}
\theoremstyle{definition}
\newtheorem{definition}[theorem]{Definition}
\newtheorem{remark}[theorem]{Remark}

\newcommand{\ie}{\emph{i.\,e.}}

\newcommand{\cf}{\emph{cf.}}

\newcommand{\abs}[1]{\left|#1\right|}
\newcommand{\norm}[1]{\left\|#1\right\|}
\newcommand{\inner}[2]{\left\langle#1,#2\right\rangle}
\newcommand{\cc}[1]{\overline{#1}}
\newcommand{\reals}{\mathbb{R}}
\newcommand{\nats}{\mathbb{N}}
\newcommand{\complex}{\mathbb{C}}

\newcommand{\cH}{\mathcal{H}}

\renewcommand\tilde{\widetilde}
\renewcommand\hat{\widehat}
\newcommand{\I}{i}

\DeclareMathOperator{\dom}{dom}
\DeclareMathOperator{\ran}{ran}

\DeclareMathOperator*{\Span}{span}
\DeclareMathOperator{\card}{card}

\DeclareMathOperator{\ind}{ind}
\DeclareMathOperator{\clos}{clos}

\begin{document}
\begin{titlepage}
\title{On a criterion for the determinate-indeterminate dichotomy of
  the moment problem
\footnotetext{%
Mathematics Subject Classification(2010):
47B32;  
47B36; 
30E05. 
}
\footnotetext{%
  Keywords:
 Hamburger moment problem;
Bases of matrix representation;
Jacobi operators.
}
\hspace{-8mm}
}
\author{
  \textbf{Diego Hernández Bustos}
\\
\small Departamento de Física Matemática\\[-1.6mm]
\small Instituto de Investigaciones en Matemáticas Aplicadas y en Sistemas\\[-1.6mm]
\small Universidad Nacional Autónoma de México\\[-1.6mm]
\small C.P. 04510, Ciudad de México\\[-1.6mm]
\small \texttt{diego.hernandez@iimas.unam.mx}
\\[2mm]
\textbf{Sergio Palafox}
\\
\small Instituto de Física y Matemáticas\\[-1.6mm]
\small Universidad Tecnológica de la Mixteca\\[-1.6mm]
\small C.P. 69000, Huajuapan, Oaxaca, México\\[-1.6mm]
\small \texttt{sergiopalafoxd@gmail.com}
\\[2mm]
\textbf{Luis O. Silva}
\\
\small Departamento de Física Matemática\\[-1.6mm]
\small Instituto de Investigaciones en Matemáticas Aplicadas y en Sistemas\\[-1.6mm]
\small Universidad Nacional Autónoma de México\\[-1.6mm]
\small C.P. 04510, Ciudad de México\\[-1.6mm]
\small \texttt{silva@iimas.unam.mx}}
\date{}
\maketitle
\vspace{-4mm}
\begin{center}
\begin{minipage}{5in}
  \centerline{{\bf Abstract}} \bigskip When the classical Hamburger
  moment problem has solutions, it has either exactly one solution or
  infinitely many solutions. Correspondingly, the moment problem is said
  to be either determinate or indeterminate. In terms of Jacobi
  operators, this dichotomy translates into the operator being either
  selfadjoint or symmetric nonselfadjoint. In this work, we present a
  new criterion for the determinate-indeterminate classification which
  hinges on bases of representation (in Akhiezer-Glazman terminology) for
  Jacobi operators so that the corresponding matrices have a certain
  structure.
\end{minipage}
\end{center}
\thispagestyle{empty}
\end{titlepage}
\section{Introduction}
\label{sec:intro}
The classical Hamburger moment problem has played a central role in
the development of modern mathematical analysis. It consists in
finding a Borel measure $\mu$ such that
\begin{equation*}
  s_{k}=\int_{\reals}t^{k}d\mu\,.
\end{equation*}
for a given real sequence of numbers $\{s_{k}\}_{k=0}^{\infty}$. To
exclude from our consideration the trivial, degenerate solutions to
the moment problem, let us always assume that the solutions have
infinite support. This deceptively simple problem leads to fundamental
questions in various fields of analysis and reveals unexpected
connections between seemingly unrelated theories and notions. When the
classical Hamburger moment problem has solutions, it has either
exactly one solution or infinitely many solutions. In the first case,
the moment problem is said to be determinate, while in the second case
it is said to be indeterminate. This dichotomy is crucial within the
moment problem theory; on the one hand, the criteria shed light on the
intricacies of the theory and interconnections between fields of
analysis, and on the other hand, the determinate and indeterminate
cases lead to two different facets of the theory.

Due to the inherent richness of the moment problem, one can approach
the determinate-indeterminate dichotomy from different viewpoints
using different mathematical notions and, consequently, there are
numerous criteria for finding out whether the moment problem is
determinate or indeterminate. There is a nonexhaustive list of these
criteria at the end of Section~\ref{sec:matr-repr-clos}, although this
is not the main point in this section, but rather the consequences of the
one-to-one correspondence between Jacobi matrices and sequences of
moments for which the corresponding moment problems admit solutions. On
the basis of this correspondence, the determinate-indeterminate
dichotomy is transformed into the selfadjoint-nonselfadjoint dichotomy
for Jacobi operators. For the passing from matrices to operators, the
concept of matrix representation for unbounded closed symmetric
operators \cite[Sec.\,47]{MR1255973} is essential.

The criterion presented in this paper is actually an if-and-only-if
criterion for resolving the selfadjoint-nonselfadjoint dichotomy for
Jacobi operators, however it does not rely on the operator theory
techniques nor on the function theoretic methods for establishing
selfadjointness or nonselfadjointness (\cf\ \cite[Chs.\,3 and
4]{MR0184042}, \cite[Ch.\,7 Sec.\,1]{MR0222718}, \cite{MR1627806} and
\cite[Ch.\,2]{MR2961944}). Instead, we use the so-called bases of
representation for Jacobi operators
(Definition~\ref{def:basis-representation}) and the result can be
stated exclusively in terms of these bases, namely:
\begin{itemize}
  \item
If for a Jacobi operator there is more than one basis of matrix
representation so that the corresponding matrix representation is a
Jacobi matrix, then the operator is selfadjoint.
\item
If for a
Jacobi operator there is only one basis of matrix representation so that
the corresponding matrix representation is a Jacobi matrix, then the
operator is nonselfadjoint.
\end{itemize}

Apart from presenting a new criterion for the
selfadjoint-nonselfadjoint dichotomy, the aim of this work is to shed
light on the relationship between sequences of moments, Jacobi
operators, matrices of representation, and measures. As a byproduct,
necessary and sufficient conditions for a basis to be a basis of
representation for a Jacobi operator are provided. Furthermore, it is
shown how to construct a basis of matrix representation so that the
corresponding measure has arbitrary index of determinacy.

Let us outline how the material of this work is presented.
Section~\ref{sec:matr-repr-clos} introduces the main objects and the
corresponding notation. This section is expository and presents
classical results on the Hamburger moment problem and its relation to
the theory of Jacobi matrices. Section~\ref{sec:self-simpl-oper} is a
review of the theory of selfadjoint simple operators and tackles the
problem of constructing bases of matrix representation for these
operators. The index of determinacy and the connection to an algorithm
to construct bases of matrix representation for selfadjoint Jacobi
operators are given in Section~\ref{sec:index-determ-bases}. Finally,
Section~\ref{sec:non-self-jacobi} deals with the case of
nonselfadjoint Jacobi operators. This section uses Krein
representation theory of symmetric operators
\cite{MR0011170,MR0011533,MR0012177,MR0048704} and de Branges theory
on Hilbert spaces of entire functions \cite{MR0229011}.

\section{Jacobi matrices and the Hamburger moment problem}
\label{sec:matr-repr-clos}
Let us introduce the notions relevant to this paper and lay out the
notation. Consider a closed symmetric operator $A$ in a Hilbert space
$\mathcal{H}$ and an orthonormal basis $\{\delta_{k}\}_{k=1}^{\infty}$
of $\mathcal{H}$. If the domain of $A$, denoted by $\dom A$,
coincides with the whole space $\mathcal{H}$ (which implies that $A$
is bounded since we have assumed it to be closed), then the operator
can be uniquely recovered from the numbers
\begin{equation}
  \label{eq:matrix-entries-from-operator}
  a_{kj}:=\inner{\delta_{k}}{A\delta_{j}}\,;
\end{equation}
here and henceforth the inner product is considered to be antilinear
in its first argument. If $\dom A\varsubsetneq\mathcal{H}$, then the
operator $A$ is not reconstructed uniquely from
\eqref{eq:matrix-entries-from-operator} even when
$\delta_{k}\in\dom A$ for any $k\in\nats$ ($\nats$ denotes the set of
positive integers). For this reason, one needs the following:
\begin{definition}
  \label{def:basis-representation}
  An orthonormal basis $\{\delta_{k}\}_{k=1}^{\infty}$ is said to be a
  basis of representation for the closed operator $A$ when
  \begin{enumerate}[(a)]
  \item $\delta_{k}\in\dom A$ for all $k\in\nats$;
  \item if there is a closed operator $B$ such that
    $B\delta_{k}=A\delta_{k}$, then $B\supset A$.
  \end{enumerate}
  When $\{\delta_{k}\}_{k=1}^{\infty}$ is a
  basis of representation for $A$, the matrix
  \begin{equation}
    \label{eq:general-infinite-matrix}
    [A]=
      \begin{pmatrix}
        a_{11} & a_{12} & a_{13}  &  a_{14}  &  \cdots\\[1mm]
        a_{21} & a_{22} & a_{23} & a_{24} & \cdots \\[1mm]
        a_{31}  &  a_{32}  & a_{33}  & a_{34} &  \\
        a_{41} & a_{42} & a_{43} & a_{44} & \ddots\\
        \vdots & \vdots &  & \ddots & \ddots
  \end{pmatrix}\,,
  \end{equation}
  with entries given by \eqref{eq:matrix-entries-from-operator}, is the
  matrix representation of $A$ with respect to
  $\{\delta_{k}\}_{k=1}^{\infty}$.
\end{definition}
 In \cite[Sec.\,47, Thm.\,3]{MR1255973}, it is established that any
 closed symmetric operator has a basis of representation. Conversely,
 if the matrix \eqref{eq:general-infinite-matrix} is Hermitian and satisfies
 \begin{equation}
   \label{eq:condition-on-matrix-for-operator}
   \sum_{j=1}^{\infty}\abs{a_{jk}}^{2}<+\infty\,, \text{ for all }k\in\nats\,,
 \end{equation}
then there is a unique closed symmetric operator $A$ such that $[A]$ is its
matrix representation with respect to a given  orthonormal basis
$\{\delta_{k}\}_{k=1}^{\infty}$ of a Hilbert space
$\mathcal{H}$.

 Let $\{q_{k}\}_{k=1}^{\infty}$ be a sequence of real numbers and
 $\{b_{k}\}_{k=1}^{\infty}$ be a sequence of positive numbers.
 An infinite matrix of the form
 \begin{equation}
   \label{eq:jacobi-matrix}
   [J]=
     \begin{pmatrix}
    q_1 & b_1 & 0  &  0  &  \cdots
\\[1mm] b_1 & q_2 & b_2 & 0 & \cdots \\[1mm]  0  &  b_2  & q_3  &
b_3 &  \\
0 & 0 & b_3 & q_4 & \ddots\\ \vdots & \vdots &  & \ddots
& \ddots
  \end{pmatrix}
 \end{equation}
 is said to be an infinite Jacobi matrix, or more specifically a
 semi-infinite Jacobi matrix to emphasize that the diagonals are
 enumerated by $\nats$ rather than $\mathbb{Z}$. Since this matrix
 satisfies \eqref{eq:condition-on-matrix-for-operator}, upon fixing an
 orthonormal basis $\{\delta_{k}\}_{k=1}^{\infty}$ of a Hilbert space
 $\mathcal{H}$, there is a unique closed symmetric operator $J$,
 called Jacobi operator, having
 $[J]$ as its matrix representation with respect to
 $\{\delta_{k}\}_{k=1}^{\infty}$. Usually, one takes
 $\mathcal{H}=l_{2}(\nats)$ and $\{\delta_{k}\}_{k=1}^{\infty}$ being
 the so-called canonical basis of $l_{2}(\nats)$, \ie\ $\delta_{k}$ is
 in turn the sequence $\{\delta_{jk}\}_{j=1}^{\infty}$, where
 $\delta_{jk}$ is the Kronecker delta. Henceforth, we assume that
 these choices for the space and the orthonormal basis are always
 made.

 Thus, the operator $J$ is the
 closure of the operator $J_{0}$ whose domain is $l_{\rm fin}(\nats)$
 (the space of sequences with a finite number of nonzero elements) and
 satisfies
\begin{equation}
  \label{eq:difference-expr}
  \begin{split}
      (J_{0} \phi)_1&:=q_1 \phi_1 + b_1 \phi_2\,,\\
     (J_{0} \phi)_k&:= b_{k-1}\phi_{k-1} + q_k \phi_k + b_k\phi_{k+1}\,,
  \quad k \in \mathbb{N} \setminus \{1\},
\end{split}
\end{equation}
for any $\phi\in l_{\rm fin}(\nats)$. Also, one verifies that
$J^{*}=J_{0}^{*}$ is the operator defined on the maximal domain, \ie,
\begin{equation}
  \label{eq:domain-maximal}
  \dom J^{*}=\{\phi\in l_{2}(\nats): \sum_{k=2}^{\infty}\abs{
    b_{k-1}\phi_{k-1} + q_k \phi_k + b_k\phi_{k+1}}^{2}<+\infty \}\,.
\end{equation}

By setting $\pi_{1}:=1$, a solution to the equations
\begin{equation}
  \label{eq:difference-equation}
  \begin{split}
      z\pi_1&:=q_1 \pi_1 + b_1 \pi_2\,,\\
     z\pi_k&:= b_{k-1}\pi_{k-1} + q_k \pi_k + b_k\pi_{k+1}\,,
  \quad k \in \mathbb{N} \setminus \{1\},\quad z\in\complex\,,
\end{split}
\end{equation}
can be found uniquely by recurrence. This solution,
$\pi(z)=\{\pi_{k}(z)\}_{k=1}^{\infty}$, is a sequence of polynomials of
$z$ called the polynomials of the first kind generated by $[J]$.
\begin{remark}
  \label{rem:on-kernel-adjoint-at-zeta}
  Since the polynomials' coefficients are real, if
  $\pi(z)\in l_{2}(\nats)$, then $\pi(\cc{z})\in l_{2}(\nats)$. Also,
  it follows from~\eqref{eq:domain-maximal} and
  \eqref{eq:difference-equation} that $\pi(z)\in l_{2}(\nats)$ if and
  only if $\pi(z)\in\ker(J^{*}-zI)$. This means on the one hand that
  the deficiency indices of the symmetric operator $J$ are always
  equal to each other, \ie\ $n_{+}(J)=n_{-}(J)$ and, on the other
  hand, if $\pi(z)\in l_{2}(\nats)$ for one nonreal $z$, then this is
  true for any nonreal $z$. When $\pi(z)\in l_{2}(\nats)$, the
  deficiency indices are equal to one because any other solution of
  \eqref{eq:difference-equation} coincides with $\pi(z)$ modulo a
  multiplicative constant (see \cite[Ch.\,4 Sec.\,1.2]{MR0184042}).
  Thus, either $n_{+}(J)=n_{-}(J)=0$ or $n_{+}(J)=n_{-}(J)=1$. Since
  $J$ is closed by definition, the case when $n_{+}(J)=n_{-}(J)=0$
  corresponds to $J$ being selfadjoint.
\end{remark}

The \emph{second order difference} expression
\eqref{eq:difference-expr} (\ie\ the matrix \eqref{eq:jacobi-matrix})
may be either in the limit point case or in the limit circle case. The
asymptotic behavior of the sequence of Weyl circles determines the
occurrence of one of these two possibilities since either the circles
degenerate into a single point or a limit circle \cite[Ch.\,1
Sec.\,3]{MR0184042}. For the class of \emph{second order differential}
expressions pertaining to the Sturm-Liouville operator, the same
dichotomy between the limit point and limit circle cases takes place
\cite[Ch.\,9]{MR0069338}. Actually, the theory behind the Weyl circles
originated in the context of differential equations.

It turns out that the limit point case corresponds to the
selfadjoint case, \ie\ $n_{+}(J)=n_{-}(J)=0$, while the limit circle
case occurs when $n_{+}(J)=n_{-}(J)=1$. This correspondence is evident
from the following expression \cite[Eq.\,1.21]{MR0184042}
\begin{equation}
  \label{eq:radii-weyl-circles}
  \left(\abs{z-\cc{z}}\sum_{k=1}^{n}\abs{\pi_{k}(z)}^{2}\right)^{-1}\,,
\end{equation}
which gives the $n$-th Weyl circle's radius for
$z\in\complex\setminus\reals$. Indeed, by von Neumann extension theory
and Remark~\ref{rem:on-kernel-adjoint-at-zeta}, selfadjointness of $J$
is equivalent to the radius vanishing as $n\to\infty$ in
\eqref{eq:radii-weyl-circles} since $\pi(z)\not\in l_{2}(\nats)$ for
$z\in\complex\setminus\reals$, while nonselfadjointness of $J$ means
that the limit of the sequence of radii \eqref{eq:radii-weyl-circles}
is not zero since, in this case, $\pi(z)\in l_{2}(\nats)$ for
$z\in\complex\setminus\reals$.

Let us now turn to the moment problem posed at the beginning of
Section~\ref{sec:intro}. A necessary and sufficient condition for a
solution to the Hamburger moment problem to exist
\cite[Thm.\,2.1.1]{MR0184042} is that
\begin{equation}
  \label{eq:positivity-of-sequence}
  \det
  \begin{pmatrix}
    s_0&s_1&\dots&s_k\\
    s_1&s_2&\dots&s_{k+1}\\
    \vdots&\vdots&\ddots&\vdots\\
    s_k&s_{k+1}&\dots&s_{2k}
  \end{pmatrix}>0
\end{equation}
for all $k\in\nats\cup\{0\}$.

For a sequence $\{s_{n}\}_{n=0}^{\infty}$ satisfying
\eqref{eq:positivity-of-sequence} there is either one solution or more
than one solution to the Hamburger moment problem. In the first case,
the moment problem is said to be determinate, while in the second
case, it is called indeterminate.

As is customary, it is assumed in this paper that the sequence of
moments $\{s_{n}\}_{n=0}^{\infty}$ is normalized, \ie\ $s_{0}=1$. This
involves no loss of generality since the general case reduces to the
normalized one by dividing the sequence of moments and its solution by
$s_{0}$.

There is a one-to-one correspondence between Jacobi matrices
\eqref{eq:jacobi-matrix} and normalized sequences $\{s_{n}\}_{n=0}^{\infty}$
satisfying \eqref{eq:positivity-of-sequence} (see
\cite[Ch.\,1]{MR0184042}). Moreover, this bijection pairs every limit
point Jacobi matrix with a sequence for which the Hamburger moment
problem is determinate and every limit circle Jacobi matrix with a
sequence for which the Hamburger moment problem is indeterminate
\cite[Thm. 2.1.2 and Cor. 2.2.4]{MR0184042}.

  Let us briefly describe how the above mentioned one-to-one
  correspondence is realized. First, consider the starting point to be
  an operator $J$ having the matrix representation
  \eqref{eq:jacobi-matrix} with respect to the orthonormal basis
  $\{\delta_{k}\}_{k=1}^{\infty}$. Since
\begin{equation*}
  \begin{split}
      J\delta_1&=q_1 \delta_1 + b_1 \delta_2\,,\\
     J\delta_k&= b_{k-1}\delta_{k-1} + q_k \delta_k + b_k\delta_{k+1}\,,
  \quad k \in \mathbb{N} \setminus \{1\}\,,
\end{split}
\end{equation*}
it is
verified that
\begin{equation}
  \label{eq:polynomials-show-simplicity}
  \delta_{k}=\pi_{k}(J)\delta_{1}\,.
\end{equation}
This means that $\delta_{1}$ is in the domain of any power of the
Jacobi operator $J$. Thus, if one defines
$s_{k-1}:=\inner{\delta_{1}}{J^{k-1}\delta_{1}}$ for all $k\in\nats$,
then a solution to the corresponding moment problem is given by the
measure
\begin{equation}
  \label{eq:extremal-measure}
  \mu(\cdot):=\inner{\delta_{1}}{E(\cdot)\delta_{1}}\,,
\end{equation}
where $E$ is either the spectral measure of $J$ if it is selfadjoint
or the spectral measure of any of the canonical selfadjoint
extensions\footnote{A canonical selfadjoint extensions of a symmetric
  operator is a selfadjoint restriction of its adjoint.} of $J$
otherwise. Hence, $\{s_{k}\}_{k=0}^{\infty}$ is a sequence of moments
and the nonselfadjoint case yields different solutions to the
corresponding moment problem. This conclusion is complemented in the
classical moment problem theory by showing, on the one hand that if
$J$ is selfadjoint, then $\mu$ is the unique solution of the moment
problem \cite[Cor.\,2.2.4]{MR0184042} and, on the other hand, that
there are other solutions apart from the ones given by the canonical
selfadjoint extensions of the nonselfadjoint Jacobi operator (see
\cite[Ch.\,2 Secs. 2 and 3]{MR0184042} and \cite[Thm.\,4]{MR1627806}).

Now, let the starting point be any normalized sequence of moments. In
this case it is known that one can construct from this sequence a
unique Jacobi matrix using the determinantal formulae (see
\cite[Ch.\,1 Sec.\,1]{MR0184042} and \cite[Thm.\,A.2]{MR1627806}). The
corresponding Jacobi operator $J$ in $l_{2}(\nats)$ turns out to be
such that $s_{k-1}=\inner{\delta_{1}}{J^{k-1}\delta_{1}}$ for all
$k\in\nats$. According to the contraposition of
\cite[Th.\,2.1.2]{MR0184042}, if the sequence of moments gives rise to
a determinate moment problem, then the Jacobi matrix is in the limit
point case. On the other hand, the contraposition of
\cite[Cor.\,2.2.4]{MR0184042} asserts that when the moment problem is
indeterminate, the Jacobi matrix is in the limit circle case.

Thus, the Hamburger moment problem is determinate if and only if the
corresponding Jacobi matrix is in the limit point case, which in turn
means that the Jacobi operator $J$ is selfadjoint. Complementarily,
the fact that the Hamburger moment problem is indeterminate is
equivalent to the corresponding Jacobi matrix being in the limit
circle case, \ie\ $J$ is not selfadjoint.

Other if-and-only-if criteria are: (a) the finite difference analogue
\cite[Thm.\,1.3.1]{MR0184042} of the Weyl alternative for
Sturm-Liouville operators \cite[Ch.\,9]{MR0069338} (related to the
limit circle/point dichotomy and the presence/absence of uniqueness of
the Weyl $m$-coefficient), (b) the Hamburger criterion (given in
terms of the moment sequence $\{s_{k}\}_{k=0}^{\infty}$) \cite[Addenda
and problems of Ch.\,2]{MR0184042}. This list is not exhaustive, but
all the criteria found in the literature boil down directly or
indirectly to the properties of the sequence $\pi(z)$.
\begin{remark}
  For the Stieltjes moment problem \cite[Pag.\,83]{MR1627806}, the
  de\-ter\-mi\-nate/in\-de\-ter\-mi\-nate dichotomy reduces to the
  existence of one/multiple nonnegative selfadjoint extensions of the
  corresponding Jacobi operator \cite[Thms.\,2 and 3.2]{MR1627806}.
  This paper is not concerned with the Stieltjes moment problem.
\end{remark}
This section concludes with an overview of the isometry map associated
with the moment problem which will be of use in the next sections.
By \cite[Thms.\,2.3.3 and 4.1.4]{MR0184042}, one has the following
classical result.
\begin{proposition}
  \label{prop:density-polynomials}
  The polynomials are dense in $L_{2}(\reals,\mu)$ if and only if
  $\mu$ is the measure given by \eqref{eq:extremal-measure}.
\end{proposition}
If one assumes that $\mu$ is given by
\eqref{eq:extremal-measure}, then the sequence of monomials
$\{t^{k-1}\}_{k=1}^{\infty}$ is total in $L_{2}(\reals,\mu)$. So, till
the end of this section, $\mu$ is assumed to be given this way. Now,
the so-called Gram-Schmidt procedure applied to
$\{t^{k-1}\}_{k=1}^{\infty}$ yields an orthonormal basis
$\{P_{k-1}(t)\}_{k=1}^{\infty}$ in $L_2(\reals,\mu)$. The Gram-Schmidt
procedure is assumed in this paper to be defined as in \cite[Ch.\,2
Sec.\,2 Thm. 5]{MR1192782} which implies that $P_{k}$ is a
polynomial of degree $k$ with positive leading coefficient. The
orthonormal sequence of polynomials is uniquely determined by these
properties \cite[Prop.\,5.1]{MR3729411}. One has the well-known
three-term relation theorem (see \cite[Sec.\,3.1.3]{MR2085852},
\cite[Prop.\,5.6]{MR3729411}, \cite[Pag.\,92]{MR1627806}):
 \begin{proposition}
   \label{prop:three-term-relation}
   If $\{P_{k-1}(t)\}_{k=1}^{\infty}$ is the orthonormal sequence of
   polynomials defined above, then
   \begin{equation}
  \label{eq:difference-equation-polynomials}
  \begin{split}
      tP_0(t)&:=q_1 P_0(t) + b_1 P_1(t)\,,\\
     tP_k(t)&:= b_{k}P_{k-1}(t) + q_{k+1} P_k(t) + b_{k+1}P_{k+1}(t)\,,
  \quad k \in \mathbb{N}\,,\quad t\in\reals\,,
\end{split}
\end{equation}
where the sequences $\{q_{k}\}_{k=1}^{\infty}$ and $\{b_{k}\}_{k=1}^{\infty}$ are
obtained from the moments $\{s_{k-1}\}_{k=1}^{\infty}$
by means of the determinantal formulae mentioned above.
\end{proposition}
\begin{remark}
  \label{rem:coincides-with-first-kind}
  The coefficients of the three-term recurrence relation
  \eqref{eq:difference-equation-polynomials} form the matrix
  \eqref{eq:jacobi-matrix}. By comparing
  \eqref{eq:difference-equation} with
  \eqref{eq:difference-equation-polynomials}, one concludes that the
  polynomials of the first kind generated by \eqref{eq:jacobi-matrix}
  coincide with the polynomials obtained by orthonormalization of
  $\{t^{k-1}\}_{k=1}^{\infty}$ in $L_{2}(\reals,\mu)$.
\end{remark}
\begin{definition}
  \label{def:map-U}
  Let $\{P_{k-1}\}_{k=1}^{\infty}$ be the orthonormal sequence of
  polynomials in $L_{2}(\reals,\mu)$ given above. Define the map
  $U:L_{2}(\reals,\mu)\to l_{2}(\nats)$ such that
\begin{equation*}
  UP_{k-1}=\delta_{k}\,,\qquad k\in\nats\,,
\end{equation*}
where $\{\delta_{k}\}_{k=1}^{\infty}$ is the canonical basis in
$l_{2}(\nats)$.
\end{definition}
By definition $U$ realizes an isometric isomorphism between
$L_{2}(\reals,\mu)$ and $l_{2}(\nats)$.
\begin{definition}
  \label{def:operator-multiplication}
  Let $\mu$ be an arbitrary Borel measure. Denote by $M_{\mu}$ the
  operator of multiplication by the independent variable in
  $L_{2}(\reals,\mu)$ defined in its maximal domain, \ie\,,
  \begin{equation*}
    \dom(M_{\mu}):=\{f\in L_{2}(\reals,\mu):
    \int_{R}t^{2}\abs{f(t)}^{2}d\mu(t)<+\infty\}\,.
  \end{equation*}
\end{definition}
Note that $M_{\mu}$ is completely determined by the measure
$\mu$. It is important to bear in mind that in the definition of
$M_{\mu}$, the measure $\mu$ is not necessarily a solution to the
moment problem. This will be relevant in the next section.

Returning to the case in which $\mu$ is a solution to a Hamburger
moment problem, if this problem is determinate, then the operator $J$ in
$l_{2}(\nats)$ whose matrix representation is \eqref{eq:jacobi-matrix}
coincides with the selfadjoint operator implicit in
\eqref{eq:extremal-measure}. In this case, on the basis of
Proposition~\ref{prop:three-term-relation}, one concludes that
$U^{-1}JU=M_{\mu}$ with $\{P_{k-1}(t)\}_{k=1}^{\infty}$ being a basis
of representation of it and \eqref{eq:jacobi-matrix} the corresponding
matrix representation. Now, if the moment problem is indeterminate,
then $J$ (whose matrix representation is \eqref{eq:jacobi-matrix}) is
not selfadjoint and the operator implicit in
\eqref{eq:extremal-measure} is a selfadjoint extension of it. Lets
denote this selfadjoint extension by $\tilde{J}$. It turns out that
$U^{-1}\tilde{J}U=M_{\mu}$, but $\{P_{k-1}(t)\}_{k\in\nats}$ is no
longer a basis of representation of it since the minimality condition
(b) of Definition~\ref{def:basis-representation} is not satisfied.
\section{Selfadjoint simple operators}
\label{sec:self-simpl-oper}
Let $A$ be a selfadjoint
operator in a separable Hilbert space $\mathcal{H}$ and $E$ be its
spectral measure given by the spectral theorem. For any real Borel set
$\partial$ and $h\in\mathcal{H}$, denote by
\begin{equation}
  \label{eq:def-mu-f}
  \mu_{h}(\partial):=\inner{h}{E(\partial)h}
\end{equation}
the corresponding nonnegative measure.
Thus, the spectral theorem allows one to define the operator
\begin{equation*}
  \phi(A):=\int_{\reals}\phi dE\,,
  \qquad\dom \phi(A):=\{h\in\mathcal{H}:\phi\in L_{2}(\reals,\mu_{h})\}\,.
\end{equation*}
\begin{definition}
  \label{def:generating-element}
  An element $g\in\mathcal{H}$ is called a generating element of the
  selfadjoint operator $A$ if the
  span over all Borel sets $\partial\subset\reals$ of $E(\partial)g$
  is dense in $\mathcal{H}$. The operator $A$ is said to be simple
  when it has a generating element.
\end{definition}
For any simple operator $A$ and any of its generating elements $g$,
there is a unitary map $\Psi_{g}$ from $L_{2}(\reals,\mu_{g})$ onto
$\mathcal{H}$ given by
\begin{equation}
  \label{eq:psy-def}
 \phi\overset{\Psi_{g}}\mapsto\phi(A)g
\end{equation}
such that the operator of multiplication $M_{\mu_{g}}$ (see
Definition~\ref{def:operator-multiplication}) is transformed into the
operator $A$. The unitary map $\Psi^{*}_{g}$ realizes the canonical
representation of the simple operator $A$ with respect to $g$.

For any Borel measure $\mu$, the operator of multiplication $M_{\mu}$
is a selfadjoint simple operator. Any function
$\eta\in L_{2}(\reals,\mu)$ such that $\eta(t)\ne 0$ for $\mu-$a.\,e.
$t$ is a generating element of $M_{\mu}$.

\begin{definition}
  \label{def:cyclic-vector}
  A vector $f$ is a cyclic vector of  $A$  when
  $f\in\dom A^{k}$ for all $k\in\nats$ and
  \begin{equation*}
    \clos\Span_{k\in\nats\cup\{0\}}A^{k}f=\mathcal{H}\,.
  \end{equation*}
\end{definition}

A cyclic vector is a generating element \cite[Sec.\,69
Thm.\,1]{MR1255973}, but the converse is not necessarily
true. However, one can always construct a cyclic vector from a
generating element. This is done below.

By
Definition~\ref{def:cyclic-vector}, a function $\eta$ is a cyclic
vector of the operator of $M_{\mu}$ if only if
\begin{equation}
    \label{eq:dense-span-cyclic}
    \clos\Span_{k\in\nats\cup\{0\}}t^{k}\eta(t)=L_{2}(\reals,\mu)\,.
  \end{equation}
  Therefore, a straightforward consequence of the canonical
  representation of simple operators is the following lemma.
\begin{lemma}
  \label{lem:creation-of-cyclic}
  Assume $\mu=\mu_{g}$, with $g$ being a generating element of a
  simple operator
  $A$. For the vector $\eta(A)g$ to be a cyclic vector of $A$ it is
  necessary and sufficient that $\eta$ satisfies
  \eqref{eq:dense-span-cyclic}.
\end{lemma}
The next statements are used to establish results pertaining to the
existence of cyclic vectors. They are
based on a reasoning used to prove
\cite[Thm.\,4.2.3]{MR0184042}. Although they are known, we present
the proofs below for the sake of completeness.
\begin{lemma}
  \label{lem:auxiliary-for-stone}
Let $\mu$ be a  $\sigma$-finite Borel measure on $\reals$ and $f$
a function in $L_{2}(\reals,\mu)$. Define
  \begin{equation}
    \label{eq:auxiliary-lemma-g}
    \mathcal{G}(t):=\int_{-\infty}^{t}f(s)d\mu(s) + C\,,
    \qquad C\in\complex\,.
  \end{equation}
  There exist a constant $C_{0}\in\complex$ such that, under the
  assumption that $C=C_{0}$ in \eqref{eq:auxiliary-lemma-g}, if
  \begin{equation}
    \label{eq:conditions-g-lemma}
    \int_{\reals}t^{k}e^{-\frac12t^{2}}d\mathcal{G}=0
  \end{equation}
  for all $k\in\nats\cup\{0\}$, then $\mathcal{G}(t)=0$ for a.\,e. t
  in $\reals$.
\end{lemma}
\begin{proof}
  If one defines
  \begin{equation*}
    C_{0}:=\frac{-1}{\sqrt{2\pi}}
\int_{\reals}\left(\int_{-\infty}^{t}f(s)
        d\mu(s)\right)e^{-t^{2}}dt\,,
  \end{equation*}
  then
  \begin{equation}
    \label{eq:first-zero}
    \int_{\reals}\mathcal{G}(t)e^{-\frac12t^{2}}dt=0\,.
  \end{equation}
  Integrating \eqref{eq:conditions-g-lemma} by parts, one arrives at
  \begin{equation}
    \label{eq:integration-by-parts}
    \int_{\reals}\left(kt^{k-1}-t^{k+1}\right)
    \mathcal{G}(t)e^{-\frac12t^{2}}dt=0\,.
  \end{equation}
  Substituting $k=0$ in this equation, one obtains
  \begin{equation}
    \label{eq:k-is-zero}
    \int_{\reals}\mathcal{G}(t)te^{-\frac12t^{2}}dt=0\,.
  \end{equation}
  Using \eqref{eq:first-zero} and \eqref{eq:k-is-zero}, it follows
  from \eqref{eq:integration-by-parts} by recurrence that
   \begin{equation}
    \label{eq:lemma-result}
    \int_{\reals}\mathcal{G}(t)t^{k}e^{-\frac12t^{2}}dt=0\,,\quad
    \forall\, k\in\nats\cup\{0\}\,.
  \end{equation}
  By the closure of the Chebyshev-Hermite functions in $L_{2}(\reals)$
  (see \cite[Thm.\,5.7.1]{MR0372517} and
  \cite[Sec.\,11.C]{MR1255973}), one concludes from
  \eqref{eq:lemma-result} that
  $\mathcal{G}(t)=0$ for a.\,e. $t\in\reals$.
\end{proof}
\begin{lemma}
  \label{lem:cyclic-every-measure}
  Let $\mu$ be an arbitrary finite Borel measure. If
  $\eta(t)=\exp(-\alpha t^{2})$ with $\alpha\ge 1/2$, then $\eta$ satisfies
  \eqref{eq:dense-span-cyclic}.
\end{lemma}
\begin{proof}
  Since $\mu$ is finite, $t^{k}\eta(t)$ is in $L_{2}(\reals,\mu)$ for
  any $k\in\nats\cup\{0\}$. Suppose that $\phi$ in $L_{2}(\reals,\mu)$
  is orthogonal to all functions $t^{k}\eta(t)$, \ie, for all
  $k\in\nats\cup\{0\}$,
  \begin{align*}
    0&=\int_{\reals}\cc{\phi(t)}t^{k}e^{-\alpha t^{2}}d\mu(t)\\
    &=\int_{\reals}t^{k}e^{-\frac12t^{2}}d\mathcal{G}(t)\,,
  \end{align*}
  where
  \begin{equation*}
    \mathcal{G}(t)=\int_{-\infty}^{t}\cc{\phi(s)}
    e^{-(\alpha-\frac12)s^{2}}d\mu(s) + C
  \end{equation*}
  with $C$ being an arbitrary constant. By
  Lemma~\ref{lem:auxiliary-for-stone}, one obtains that
  $\mathcal{G}(t)=0$ for a.\,e. $t\in\reals$.  Thus,
  \begin{equation*}
  \norm{\phi}^{2}=
  \int_{\reals}e^{(\alpha-\frac12)t^{2}}\phi(t)d\mathcal{G}(t)=0\,.
\end{equation*}
\end{proof}

The conclusion of Lemma~\ref{lem:cyclic-every-measure} motivates the
following definition:
\begin{definition}
  \label{def:stone-vectors}
  Let $g$ be a generating element for the selfadjoint operator
  $A$. For any $\alpha\ge 0$, define
  \begin{equation*}
    \eta(\alpha,g):=\exp(-\alpha A^{2})g\,.
  \end{equation*}
  We refer to $\eta(\alpha,g)$ as to the Stone vector of order
  $\alpha$ obtained from the generating element $g$.
\end{definition}
The combination of Lemmas~\ref{lem:creation-of-cyclic} and
\ref{lem:cyclic-every-measure} yields the following assertion which is
the first part of a slight generalization of Stone classical result
(see \cite[Thm.\, 4.2.3]{MR0184042}).
\begin{corollary}
  \label{cor:stonecyclic}
  For any generating element $g$ of a simple selfadjoint operator
  $A$, any Stone vector $\eta(\alpha,g)$ is a cyclic vector of $A$
  for all $\alpha\ge 1/2$.
\end{corollary}

\begin{remark}
  \label{rem:delta-1-cyclic}
  Let $J$ be the operator whose matrix representation is
  \eqref{eq:jacobi-matrix} with respect to the canonical basis
  $\{\delta_{k}\}_{k=1}^{\infty}$. If $J$ is selfadjoint, then it
  follows from \eqref{eq:polynomials-show-simplicity} that $J$ is
  simple and $\delta_{1}$ is a cyclic vector of it. If
  $J\varsubsetneq J^{*}$, then $\delta_{1}$ is a cyclic vector for
  each of the selfadjoint extensions of $J$ (and therefore each
  selfadjoint extension is simple).
\end{remark}
The next proposition amounts, in a certain sense, to the converse of
the assertion in the preceding remark.

\begin{proposition}
  \label{prop:gram-schmidt-basis}
  Let $\delta$ be a cyclic vector of $A$ and
  $\{\delta_{k}\}_{k=1}^{\infty}$ be the
  orthonormal basis obtained from applying the Gram-Schmidt procedure
  to the sequence $\{A^{k-1}\delta\}_{k=1}^{\infty}$. If
  $B$ is the minimal closed operator such that
  $B\delta_{k}=A\delta_{k}$ for all $k\in\nats$ (\cf\
  Definition~\ref{def:basis-representation}), then the matrix
  representation of
  $B$ with respect to $\{\delta_{k}\}_{k=1}^{\infty}$ is a
  semi-infinite Jacobi matrix (see \eqref{eq:jacobi-matrix}).
\end{proposition}
\begin{proof}
  First note that, by the Gram-Schmidt algorithm (see \cite[Ch.\,2
Sec.\,2 Thm. 5]{MR1192782}), one has
$\delta_{1}=\delta/\norm{\delta}$, so $\delta_{1}$ is a normalized
cyclic vector. Taking this into account, the proof reduces to a well known assertion
  \cite[Sec.\,3.1.3]{MR2085852} on orthogonal polynomials by means of the
  canonical representation of $A$ with respect to $\delta_{1}$. Indeed,
  using the map introduced in \eqref{eq:psy-def}, one has
  \begin{equation*}
   \Psi^{*}_{\delta_{1}}(A^{k-1}\delta_{1})=t^{k-1}
 \end{equation*}
 for any $k\in\nats$. The unitarity of the map given in
 \eqref{eq:psy-def} and Definition~\ref{def:cyclic-vector} imply that
 the sequence $\{t^{k-1}\}_{k=1}^{\infty}$ is total in
 $L_2(\reals,\mu_{\delta_{1}})$. Therefore, the Gram-Schmidt procedure
 applied to $\{t^{k-1}\}_{k=1}^{\infty}$ yields an orthonormal basis
 $\{P_{k-1}(t)\}_{k\in\nats}$ in $L_2(\reals,\mu_{\delta_{1}})$ (see the
 end of Section~\ref{sec:matr-repr-clos}). Clearly,
 $\delta_{k}=\Psi_{\delta_{1}}P_{k-1}$. Therefore, on the basis of
 Proposition~\ref{prop:three-term-relation}, one concludes that the
 numbers
 \begin{equation}
  \label{eq:obtained-matrix}
  a_{jk}:=\inner{\delta_{j}}{A\delta_{k}}=
  \inner{P_{j-1}}{tP_{k-1}}_{L_{2}(\reals,\mu_{\delta_{1}})}\,,\quad
j,k\in\nats\,,
\end{equation}
generate a semi-infinite Jacobi matrix. To
finish the proof notice that, by
Definition~\ref{def:basis-representation}, $B$ is the operator
whose matrix representation has the entries
\eqref{eq:obtained-matrix}.
\end{proof}
\begin{remark}
  \label{rem:example}
  In the assertion of Proposition~\ref{prop:gram-schmidt-basis}, it
  could be that $B\varsubsetneq A$, \ie\ the orthonormal basis
  obtained from the Gram-Schmidt procedure applied to
  $\{A^{k-1}\delta\}_{k=1}^{\infty}$ is not necessarily a basis of
  representation for $A$. An example of this has already appeared at
  the end of Section~\ref{sec:matr-repr-clos}. Indeed, let $J$ be the
  operator whose matrix representation is \eqref{eq:jacobi-matrix}
  with respect to the canonical basis $\{\delta_{k}\}_{k=1}^{\infty}$
  and assume that $J\varsubsetneq J^{*}$. By
  Remark~\ref{rem:delta-1-cyclic}, $\delta_{1}$ is a cyclic vector of
  $\tilde{J}$, a fixed selfadjoint extension of $J$. Moreover, it
  follows from \eqref{eq:polynomials-show-simplicity} and
  Remark~\ref{rem:coincides-with-first-kind} that the basis
  $\{\delta_{k}\}_{k=1}^{\infty}$ is obtained from the Gram-Schmidt
  procedure applied to $\{\tilde{J}^{k}\delta_{1}\}$ since
  $\tilde{J}\supset J$. Note that $\{\delta_{k}\}_{k=1}^{\infty}$ is
  the basis of representation for $J$, but not for $\tilde{J}$.
\end{remark}
The next assertion is a slight generalization of a classical result by
Stone on simple operators (see \cite[Thm.\,4.2.3]{MR0184042}).
\begin{proposition}
  \label{prop:basis-representation}
  For any simple selfadjoint operator, there is an uncountable set of
  bases of matrix representation such that the corresponding matrix
  representation of the operator with respect to each of the bases is
  a Jacobi matrix.
\end{proposition}
\begin{proof}
 Let $A$ be a simple operator and $g$ a generating element of it. If
 $\alpha\ge\frac12$, then $\eta(\alpha,g)$ given in
 Definition~\ref{def:stone-vectors} is a cyclic vector of $A$. Due to
 Proposition~\ref{prop:gram-schmidt-basis}, if
 $\{\delta_{k}(\alpha,g)\}_{k=1}^{\infty}$ is the orthonormal basis obtained
 from applying the Gram-Schmidt procedure to the sequence
 $\{A^{k-1}\eta(\alpha,g)\}_{k=1}^{\infty}$, then
 \begin{equation*}
    \inner{\delta_{j}(\alpha,g)}{A\delta_{k}(\alpha,g)}
 \end{equation*}
is a Jacobi matrix which will be denoted by $[A](\alpha,g)$.

It remains to prove that $[A](\alpha,g)$ is the matrix representation
of $A$ with respect to the orthonormal basis
$\{\delta_{k}(\alpha,g)\}_{k=1}^{\infty}$.  According to
Definition~\ref{def:basis-representation}, this boils down to showing
that $A$ is the minimal closed operator associated with the matrix
$[A](\alpha,g)$.

Let $B$ be the operator whose matrix representation is $[A](\alpha,g)$
(on account of what is said in the paragraph below
\eqref{eq:jacobi-matrix} such operator is univocally determined by the
matrix and this operator is symmetric). Assume that $h:=\phi(A)g$ is
orthogonal to $(B-\I I)\delta_{k}(\alpha,g)$ for all
$k\in\nats$, then
\begin{align*}
  0&=\inner{h}{(B-\I I)\delta_{k}(\alpha,g)}
  =\inner{h}{(A-\I I)\delta_{k}(\alpha,g)}\\
   &=\inner{\phi(A)g}{(A-\I I)P_{k-1}(A)e^{-\alpha A^{2}}g}\\
  & =\int_{\reals}\cc{\phi(t)}(t-\I)
    P_{k-1}(t)e^{-\alpha t^{2}}d\mu_{g}(t)\,,
\end{align*}
where the second equality holds since $B\subset A$. In the third
equality, one uses Definition \ref{def:stone-vectors} and
the fact that $\delta_{k}(\alpha,g)=\Psi_{\eta(\alpha,g)}(P_{k-1})$
(see the proof of Proposition~\ref{prop:gram-schmidt-basis}). In the last equality, one
recurs to the isometric property of $\Psi_{g}$.  Thus, for any
$k\in\nats$, one has
\begin{align*}
  0&=\int_{\reals}\cc{\phi(t)}(t-\I)
     t^{k-1}e^{-\alpha t^{2}}d\mu_{g}(t)\\
  &=\int_{\reals}t^{k}e^{-\frac12t^{2}}d\mathcal{F}\,,
\end{align*}
where
\begin{equation*}
  \mathcal{F}(t):=\int_{-\infty}^{t}\cc{\phi(t)}(s-\I)
  e^{-(\alpha-\frac12)s^{2}}d\mu_{g}(s)+ C
\end{equation*}
with $C$ being an arbitrary constant. By
Lemma~\ref{lem:auxiliary-for-stone}, $\mathcal{F}(t)=0$ for
a.\,e. $t\in\reals$. Therefore
\begin{align*}
  0&=\int_{\reals}(t-\I)\cc{\phi(t)}d\mathcal{F}\\
  &=\int_{\reals}\abs{t+\I}^{2}e^{-\frac12t^{2}}\abs{\phi(t)}^{2}d\mu_{g}(t)\,.
\end{align*}
This implies that $\norm{\phi}_{L_{2}(\reals,\mu_{g})}=0$.  Thus, one
concludes that the deficiency space of $B$ on the upper-half plane
is trivial and therefore $B$ is maximal which, in turn, means that it
does not have proper symmetric extensions.
\end{proof}
\begin{remark}
  \label{rem:special-cyclic}
  For any generating element $g$ of $A$ and $\alpha\ge\frac12$, the
  definition of Stone vectors $\eta(\alpha,g)$ by means of a Gaussian
  function guarantees not only cyclicity, but also the fact that
  $\{\delta_{k}(\alpha,g)\}_{k=1}^{\infty}$ is a basis of representation for $A$
  (\cf\ Remark~\ref{rem:example}).
\end{remark}
The following assertion gives necessary and sufficient conditions for a
cyclic vector $\delta$ of $A$ to generate, through the
Gram-Schmidt procedure applied to the sequence
$\{A^{k-1}\delta\}_{k=1}^{\infty}$, a basis of representation for $A$.
\begin{proposition}
  \label{prop:necessary-and-sufficient}
  Let $A$ be a simple operator and $\delta$ a cyclic vector of it. The
  Gram-Schmidt procedure applied to the sequence
  $\{A^{k-1}\delta\}_{k=1}^{\infty}$ yields a basis of representation
  for $A$ if and only if
 \begin{equation}
    \label{eq:special-cyclic}
    \clos\Span_{k\in\nats}\{(A-\I I)A^{k-1}\delta\}=\mathcal{H}\,.
  \end{equation}
\end{proposition}
\begin{proof}
  Let $\{\delta_{k}\}_{k=1}^{\infty}$ be the orthonormal basis
  obtained by the Gram-Schmidt procedure applied to the sequence
  $\{A^{k-1}\delta\}_{k=1}^{\infty}$. Thus,
  $\delta_{k}=P_{k-1}(A)\delta$, where $P_{k}$ is a polynomial of
  degree $k$ (see the proof of
  Proposition~\ref{prop:gram-schmidt-basis}). Denote by $B$ the
  minimal closed operator so that $B\delta_{k}=A\delta_{k}$. 
  Assume first that \eqref{eq:special-cyclic}
  holds. If the vector $h$ is such that $\inner{h}{(B-\I
    I)\delta_{k}}$ vanishes for all $k\in\nats$, then
  \begin{equation}
    \label{eq:sufficiency-aux}
    \inner{h}{(A-\I I)A^{k-1}\delta}=0\,,\quad \forall\,k\in\nats\,.
  \end{equation}
  Therefore, it follows from \eqref{eq:special-cyclic} and
  \eqref{eq:sufficiency-aux} that $h=0$. Since $\ran(B-\I
  I)$ contains $\Span_{k\in\nats}(B-\I I)\delta_{k}$, one concludes that
  the closed symmetric operator $B$ is maximal and therefore $B=A$.

  Now suppose that $B=A$ and \eqref{eq:special-cyclic} does
  not hold, \ie\ there is a nonzero vector $h$ so that $h\perp (A-\I
  I)A^{k-1}\delta$ for all $k\in\nats$. This implies that
  \begin{equation}
    \label{eq:condition-h-delta-k}
    \inner{h}{(B-\I I)\delta_{k}}=0\,,\quad \forall\,k\in\nats\,,
  \end{equation}
  since $\delta_{k}$ is a polynomial of $A$ applied to $\delta$. By
  Proposition~\ref{prop:gram-schmidt-basis}, $B$ is a Jacobi
  operator so one can denote the entries of the corresponding matrix
  as in \eqref{eq:jacobi-matrix}. Therefore, by writing
  $h=\sum_{k=1}^{\infty}h_{k}\delta_{k}$, one obtains from
  \eqref{eq:condition-h-delta-k} that
  \begin{equation*}
  \begin{split}
      \I h_1&:=q_1 h_1 + b_1 h_2\,,\\
     \I h_k&:= b_{k-1}h_{k-1} + q_k h_k + b_kh_{k+1}\,,
  \quad k \in \mathbb{N} \setminus \{1\}\,.
\end{split}
\end{equation*}
By the assumption that $B=A$, $B$ is selfadjoint and
therefore
\begin{equation*}
  \sum_{k=1}^{\infty}\abs{h_{k}}^{2}=+\infty\,.
\end{equation*}
This contradicts the fact that $h$ is a nonzero element of the space.
\end{proof}
To close up this section, we put its results in the context of Jacobi
operators in the selfadjoint and nonselfadjoint cases.

A straightforward consequence of
Proposition~\ref{prop:basis-representation} and
Remark~\ref{rem:delta-1-cyclic} is the following:
\begin{proposition}
  \label{prop:representations-jacobi-operator}
  Let $J$ be the operator whose matrix representation is
  \eqref{eq:jacobi-matrix} with respect to the canonical basis
  $\{\delta_{k}\}_{k=1}^{\infty}$. If \eqref{eq:jacobi-matrix} is in
  the limit point case, then, for any $\alpha\ge \frac12$,
  the basis $\{\delta_{k}(\alpha,g)\}_{k=1}^{\infty}$, constructed
  from any generating element $g$, is a basis of
  matrix representation for $J$ and the corresponding matrix
  $[J](\alpha,g)$ is a Jacobi matrix in the limit point case.
\end{proposition}
\begin{remark}
  \label{rem:relation-moments-matrices}
  If the hypotheses of the preceding proposition hold, then, for each
  generating element $g$ of $J$ and $\alpha\ge\frac12$, the Jacobi
  matrix $[J](\alpha,g)$ generates a sequence of moments
  $\{s_{k}(\alpha,g)\}_{k=0}^{\infty}$ so that the solution to the
  corresponding moment problem is unique. By \cite[Ch.\,5 Sec.\,3
  Lem.\,3]{MR1192782}, this unique solution is given by
  \begin{equation}
    \label{eq:sol-moment-alpha}
    \mu_{\eta(\alpha,g)}(\partial)=\int_{\partial}e^{-2\alpha
      t^{2}}d\mu_{g}(t)
  \end{equation}
  for any Borel set $\partial$.

  According to Remark~\ref{rem:delta-1-cyclic}, if $J$ is a
  nonselfadjoint Jacobi operator, then, for each selfadjoint extension
  $\tilde J$ and $\alpha\ge\frac12$, the Jacobi
  matrix $[\tilde J](\alpha,\delta_{1})$ is associated with a
  determinate moment problem whose unique solution is
\begin{equation}
    \label{eq:sol-moment-alpha}
    \mu_{\eta(\alpha,\delta_{1})}(\partial)=\int_{\partial}e^{-2\alpha
      t^{2}}d\tilde\mu_{\delta_{1}}(t)
  \end{equation}
  for any Borel set $\partial$, where $\tilde\mu_{\delta_{1}}$ is
  given by \eqref{eq:def-mu-f} with $h=\delta_{1}$ and $E$ being the
  spectral measure of $\tilde J$.
  \end{remark}
\section{Index of determinacy and bases of representation}
\label{sec:index-determ-bases}

In the previous section, bases of matrix representation for a
selfadjoint Jacobi operator $J$ were constructed from an arbitrary
generating element of it by means of the Stone vectors
(Definition~\ref{def:stone-vectors}). The matrices representing $J$
with respect to these bases were Jacobi matrices. The function
involved in Definition~\ref{def:stone-vectors} guarantee not only that
the Stone vector is a cyclic vector, but also that the basis obtained
from it is a basis of matrix representation.

In this section, an alternative method is used for the construction of
bases of matrix representation of a selfadjoint Jacobi operator so
that the corresponding matrix is a Jacobi matrix. This method is
related to the so-called index of determinacy \cite{MR1308001} of a
solution to the moment problem.

Let $\mu$ be a Borel measure. For any Borel set $\partial$, denote
 \begin{equation}
    \label{eq:leser-determinacy-measure}
    \mu_{n}(\partial):=\int_{\partial}(1+x^{2})^{n}d\mu(x)\,,\quad n\in\nats\,.
  \end{equation}
Note that $\mu_{n}$ is obtained by applying the transformation
\eqref{eq:leser-determinacy-measure} with $n=1$ to the measure
$\mu_{n-1}$.

Two classical results pertaining to the density of
polynomials in $L_{2}$ spaces are Propositions
\ref{prop:density-polynomials} and
\begin{proposition}
  \label{prop:riez-index}
  The measure $\mu$ is the solution to a determinate Hamburger moment problem if
  and only if the polynomials are dense in  $L_{2}(\reals,\mu_{1})$.
\end{proposition}
The proof of this assertion is found in \cite{mriesz1923} (see also
\cite[Lem.\,A]{MR1308001} and \cite[Cor.\,6.11]{MR3729411}).

\begin{definition}
  \label{def:index-determinacy}
  The index of determinacy of a solution $\mu$ to a Hamburger moment
  problem is
  \begin{equation*}
    \ind\mu:=\sup\{n\in\nats: \text{ the polynomials are dense in }
    L_{2}(\reals,\mu_{n})\}.
  \end{equation*}
  It is not excluded that $\ind\mu$ could be $\infty$, which takes
  place when the polynomials are dense in $L_{2}(\reals,\mu_{n})$ for
  any $n\in\nats$.
\end{definition}
This definition differs from the one in \cite{MR1308001}. If $\mu$ has
index $n$ according to \cite[Eq.\,1.1]{MR1308001}, then $\ind\mu=n+1$
by Definition~\ref{def:index-determinacy}. The index of determinacy in
this paper is so that any determinate measure has positive index of
determinacy. Indeed, by Proposition~\ref{prop:riez-index}, the index
of determinacy makes sense only for solutions to determinate moment
problems and for any such solution $\mu$, $\ind\mu\ge1$. Note that the
index of determinacy decreases one unit each time the transformation
\eqref{eq:leser-determinacy-measure} with $n=1$ is applied. Also, it
follows from Propositions \ref{prop:density-polynomials} and
\ref{prop:riez-index} that if $\ind\mu=1$, then $\mu_{1}$ is given by
\eqref{eq:extremal-measure} with $E$ being the spectral measure of a
canonical selfadjoint extension of a nonselfadjoint Jacobi operator
and $\mu_{2}$ is such that the polynomials are no longer dense in
$L_{2}(\reals,\mu_{2})$.

By reverting the transformation \eqref{eq:leser-determinacy-measure},
one can increase the index of determination of a given measure.
Indeed, let $J$ be a nonselfadjoint Jacobi operator in $l_{2}(\nats)$
and $[J]$ its matrix representation with respect to the canonical
basis $\{\delta_{k}\}_{k=1}^{\infty}$. Fix a canonical selfadjoint
extension $\tilde J$ of $J$ and denote by $\mu$ the measure given by
\eqref{eq:extremal-measure} with $E$ being the spectral measure of
$\tilde J$. Now, for any
Borel set $\partial$ and $n\in\nats$, define
\begin{equation*}
  \nu_{n}(\partial):=\frac{1}{C}\int_{\partial}(1+x^{2})^{-n}d\mu(x)\,,\quad{
  where }\quad C:=
  \int_{\reals}(1+x^{2})^{-n}d\mu(x)\,.
\end{equation*}
Due to Propositions~\ref{prop:density-polynomials} and
\ref{prop:riez-index}, one verifies that according to
Definition~\ref{def:index-determinacy} $\ind\nu_{n}=n$ for any
$n\in\nats$.

Let $\{R_{k-1}\}_{k=1}^{\infty}$ be the orthonormal
sequence of polynomials in $L_{2}(\reals,\nu_{1})$ obtained from
monomials by the Gram-Schmidt procedure. Define
\begin{equation}
  \label{eq:orthonormal-functions}
  f_{k-1}(t):=\frac{1}{\sqrt{C}(t-\I)}R_{k-1}(t)\,,\quad
  \text{ for all }\quad k\in\nats\,.
\end{equation}
Thus,
\begin{equation}
  \label{eq:orthonormality-f}
  \inner{f_{k}}{f_{j}}_{L_{2}(\reals,\mu)}=
  \frac{1}{C}\int_{\reals}(1+t^{2})^{-1}\cc{R_{k}(t)}R_{j}(t)d\mu=
  \int_{\reals}\cc{R_{k}(t)}R_{j}(t)d\nu_{1}=\delta_{jk}
\end{equation}
so that $\{f_{k-1}\}_{k=1}^{\infty}$ is orthonormal in
$L_{2}(\reals,\mu)$. This orthonormal system is also complete due to
the fact that if, for any $k\in\nats$,
\begin{equation*}
  0=\inner{f_{k-1}}{h}_{L_{2}(\reals,\mu)}=
  \frac{1}{\sqrt{C}}\int_{\reals}(t-\I)^{-1}R_{k-1}(t)h(t)d\mu(t)\,,
\end{equation*}
then $h(t)=0$ for $\mu$-a.\,e. $t\in\reals$ since the polynomials are
dense in $L_{2}(\reals,\mu)$ and $(t-\I)^{-1}$ never vanishes on the
real line. It is equally straightforward to establish that $f_{k-1}$ is
in the domain of the multiplication operator $M_{\mu}$ (see
Definition~\ref{def:operator-multiplication}) for any $k\in\nats$ since $R_{k-1}$ is in the
domain of $M_{\nu_{1}}$ for any $k\in\nats$.

Thus, the orthonormal basis $\{f_{k-1}\}_{k=1}^{\infty}$ in
$L_{2}(\reals,\mu)$ satisfies (a) of
Definition~\ref{def:basis-representation} with respect to the operator
of multiplication $M_{\mu}$.

\begin{proposition}
  \label{prop:f-is-basis-representation}
  If $\mu$ is a solution to an indeterminate moment problem, then the
  orthonormal basis $\{f_{k-1}\}_{k=1}^{\infty}$ defined above is a
  basis of matrix representation for the operator $M_{\mu}$.
\end{proposition}
\begin{proof}
  Item (a) of Definition~\ref{def:basis-representation} has already
  been established. Let us show that the sequence
  $\{f_{k-1}\}_{k=1}^{\infty}$ is obtained by the Gram-Schmidt
  procedure applied to the sequence
  $\{t^{k-1}(\sqrt{C}(t-\I))^{-1}\}_{k=1}^{\infty}$ in
  $L_{2}(\reals,\mu)$. Indeed, proceeding as in
  \eqref{eq:orthonormality-f}, one has for the first step of the
  Gram-Schmidt algorithm \cite[Ch.\,2 Sec.\,2 Thm.\,5]{MR1192782}
  \begin{align}
    &(\sqrt{C}(t-\I))^{-1}[t-
      \inner{t(\sqrt{C}(t-\I))^{-1}}{(\sqrt{C}(t-\I))^{-1}}_{L_{2}(\reals,\mu)}]
    \nonumber\\
    =&(\sqrt{C}(t-\I))^{-1}[t-
      \inner{t}{1}_{L_{2}(\reals,\nu_{1})}]\,.\label{eq:first-step-gs}
  \end{align}
  The expression in the square brackets of \eqref{eq:first-step-gs} is
  the first step of Gram-Schmidt procedure applied to the sequence
  $\{t^{k-1}\}_{k=1}^{\infty}$ in $L_{2}(\reals,\nu_{1})$. By
  induction, taking into account \eqref{eq:orthonormal-functions}, one
  verifies that $\{f_{k-1}\}_{k=1}^{\infty}$ is the result of
  orthonormalizing the sequence
  $\{t^{k-1}(\sqrt{C}(t-\I))^{-1}\}_{k=1}^{\infty}$. In particular,
  this shows that $f_{0}=(\sqrt{C}(t-\I))^{-1}$ is a cyclic vector of
  $M_{\mu}$ since $\{f_{k-1}\}_{k=1}^{\infty}$ is total in $L_{2}(\reals,\mu)$.

  Now, for any $k\in\nats$, one has
  \begin{equation*}
    [(M_{\mu}-\I I) M_{\mu}^{k-1}f_{0}](t)=t^{k-1}\,,
  \end{equation*}
  which, on the basis of Proposition~\ref{prop:density-polynomials},
  implies that \eqref{eq:special-cyclic} holds for $M_{\mu}$. Thus,
  Proposition~\ref{prop:necessary-and-sufficient} leads to the desired
  conclusion.
\end{proof}
  As a consequence of Propositions~\ref{prop:gram-schmidt-basis} and
  \ref{prop:f-is-basis-representation}, the matrix representation of
  $M_{\mu}$ with respect to $\{f_{k-1}\}_{k=1}^{\infty}$ is a Jacobi
  matrix. This matrix can be found by observing that the sequence
  $\{f_{k-1}\}_{k=1}^{\infty}$ satisfies the same three-term
  recurrence relation that the sequence of polynomials
  $\{R_{k-1}\}_{k=1}^{\infty}$ does. Since $\ind\nu_{1}=1$, the
  coefficients of the recurrence relation form a Jacobi matrix in the
  limit point case, which is denoted by $[\hat J]$. Hence, the matrix
  representation of $M_{\mu}$ with respect to
  $\{f_{k-1}\}_{k=1}^{\infty}$ is $[\hat J]$.
  Note that $U^{-1}\tilde JU=M_{\mu}$, where $U$ is the map given in
  Definition~\ref{def:map-U}. Thus, if one defines
  \begin{equation*}
    \omega_{k}:=Uf_{k-1}\,,\quad\text{ for all}\quad k\in\nats\,,
  \end{equation*}
  then $\{\omega_{k}\}_{k=1}^{\infty}$ is a basis of matrix
  representation for $\tilde J$ and the corresponding matrix is
  $[\hat J]$. As has been said before (see Remark~\ref{rem:example}),
  $\{\delta_{k}\}_{k=1}^{\infty}$ is not a basis of representation for
  $\tilde J$. Note that the sequence $\{\omega_{k}\}_{k=1}^{\infty}$ is
  not in $\dom J$ since otherwise the minimality condition (b) of
  Definition~\ref{def:basis-representation} for $\tilde J$ is
  violated.

It is worth mentioning that the measure $\nu_{1}$, which gives rise to the matrix
$[\hat J]$, has the smallest index that a determinate measure could
have. Furthermore, for any $n\in\nats$, by using $\nu_{n}$ and
modifying accordingly \eqref{eq:orthonormal-functions}, one can
construct a basis of matrix representation for $\tilde
J$ so that the corresponding matrix is a Jacobi matrix.

In contrast to the construction given above, the measures appearing
in Section~\ref{sec:self-simpl-oper} have infinite index of
determinacy. This is asserted in the following proposition.
\begin{proposition}
  \label{prop:infinite-index}
  For any $\alpha>0$, the measure $\mu_{\eta(\alpha,g)}$ given in
  \eqref{eq:sol-moment-alpha} has infinite index of determinacy.
\end{proposition}
\begin{proof}
  Denote by $\chi_{t}$ the measure
  \begin{equation*}
    \chi_{t}(\partial):=
    \begin{cases}
      1 & t\in\partial\\
      0 & t\not\in\partial\,,
    \end{cases}
  \end{equation*}
  where $\partial\subset\reals$ is a Borel set.

  If one assumes that $\mu_{\eta(\alpha,g)}$ has finite index, then
  the measure is discrete \cite[Cor.\,3.4]{MR1308001} and according to
  \cite[Thm.\, 3.9]{MR1308001}
  there is a finite collection of real numbers $t_{1},\dots,t_{n}$ out of the support of
  $\mu_{\eta(\alpha,g)}$ so that, for any positive numbers $a_{1},\dots,a_{n}$,
  \begin{equation*}
    \tilde\mu := \mu_{\eta(\alpha,g)}+\sum_{k=1}^{n}a_{k}\chi_{t_{k}}
  \end{equation*}
  is indeterminate. But, one verifies
  \begin{equation*}
  \int_{\reals}e^{\abs{t}}d\tilde\mu(t)=\int_{\reals}e^{\abs{t}-2\alpha
  t^{2}}d\mu_{g}(t)+\sum_{k=1}^{n}a_{k}e^{\abs{t_{k}}}\chi_{t_{k}}<+\infty\,.
\end{equation*}
Whence, by \cite[Thm.\,5.2]{MR0481888}, one concludes that $\tilde\mu$
is determinate which is a contradiction.
\end{proof}

\section{Non-selfadjoint Jacobi operators}
\label{sec:non-self-jacobi}

This section begins with an account on some of the remarkable
properties of the class of entire operators \cite{MR0012177} to which
the class of nonselfadjoint Jacobi operators belongs. One of these
properties leads to the fact that a nonselfadjoint Jacobi operator has
a unique matrix representation being a Jacobi matrix.

\begin{definition}
  \label{def:regularity}
  A closed operator $A$ in a Hilbert space $\cH$ is said to be
  regular when for any $z\in\complex$, there is a constant $C>0$
  (which could depend on $z$) such that
\begin{equation*}
  \norm{(A-zI)\phi}\ge C\norm{\phi}
\end{equation*}
for all $\phi\in\dom A$.
\end{definition}
The fact that an operator is regular means that the its spectral
kernel is empty, therefore every regular symmetric operator is
completely nonselfadjoint (\ie\ there is no invariant subspace of the
operator in which it induces a selfadjoint operator). Indeed, since
any part of an operator with empty spectral kernel has empty spectral
kernel, this part cannot be selfadjoint. It is noteworthy that there
are completely nonselfadjoint operators which are not regular.

Completely
nonselfadjointness of a closed symmetric operator $A$ means that
\cite[Thm. 1.2.1]{MR1764951}
\begin{equation}
  \label{eq:completely-nonselfadjoint}
  \bigcap_{z\in\complex\setminus\reals}\ran(A-z I)=\{0\}\,,
\end{equation}
since the l.h.s of \eqref{eq:completely-nonselfadjoint} is the maximal
invariant subspace in which $A$ is selfadjoint
\cite[Thm.\,1.2.1]{MR1764951}.

\begin{definition}
  \label{def:involutive-conjugation}
  An antilinear map $\mathcal{I}$ of $\mathcal{H}$ onto itself being
  an involution (\ie\ $\mathcal{I}^{2}=\mathcal{J}$) and such that
  \begin{equation}
    \label{eq:inversion-inner-conjugation}
    \inner{\mathcal{I}\phi}{\mathcal{I}\psi}=\inner{\psi}{\phi}\quad\text{
    for any }\phi,\psi\in\mathcal{H}
  \end{equation}
  is called a \emph{conjugation} (see \cite[Sec.\,13.1]{MR2953553} and
  \cite[Eq.\,8.1]{MR566954}).
\end{definition}
A procedure for constructing a conjugation commuting (see
\cite[Eq.\,8.1]{MR566954}) with all canonical selfadjoint extensions
of a symmetric completely nonselfadjoint operator $A$ with
one-dimensional deficiency spaces is presented in
\cite[Prop.\,2.3]{MR3002855}. This conjugation commutes with $A$ (see
the proof of \cite[Prop.\,13.25(ii)]{MR2953553}). Conversely, if a
conjugation commutes with a symmetric operator $A$ with deficiency
indices $n_{+}(A)=n_{-}(A)=1$, then the conjugation commutes with all
canonical selfadjoint extensions of $A$ (see
\cite[Prop.\,13.25 (iv)]{MR2953553} and \cite[Cor.\,2.5]{MR1627806}).

The following statement is motivated by Krein's representation theory
of symmetric operators \cite{MR0011170,MR0011533,MR0012177,MR0048704}.
The assertion's constructive proof can be found in
\cite[Prop.\,2.12]{MR3002855}.
\begin{proposition}
  \label{prop:function-xi}
  If $A$ is a regular, symmetric operator such that
  $n_{+}(A)=n_{-}(A)=1$ and $\mathcal{I}$ is a conjugation that
  commutes with $A$, then there is a vector function
  $\xi_{A}:\complex\to\mathcal{H}$ with the following properties:
\begin{enumerate}[(a)]
\item $\xi_{A}$ is entire and zero-free.
\item $\xi_{A}(z)\in\ker(A^{*}-z I)$ for each $z\in\complex$.
\item For all $z\in\complex$, $\mathcal{I}\xi_{A}(z)=\xi_{A}(\cc{z})$.
\end{enumerate}
\end{proposition}
Having fixed the involution $\mathcal{I}$, the function $\xi_{A}$ is
uniquely determined modulo a multiplicative scalar factor being an
entire, zero-free function which turns out to be real (see
\cite[Rem.\,2.13]{MR3002855} and \cite[Lem.\,3]{zbMATH06526214}).
Recall that a complex valued function $f$ of complex variable
satisfying
\begin{equation}
  \label{eq:real-function}
  \cc{f(z)}=f(\cc{z})
\end{equation}
is called \emph{real}; thus a real function is real on the real line.
\begin{remark}
  \label{rem:function-xi}
  It is worth mentioning that if, for a closed symmetric operator $A$
  with $n_{+}(A)=n_{-}(A)=1$, the equality
  \eqref{eq:completely-nonselfadjoint} holds and there is a function
  $\xi_{A}$ satisfying (a)--(c) of Proposition~\ref{prop:function-xi},
  then the operator is regular. This is proven by means of the
  functional model given in \cite[Sec.\,2.3]{MR3002855} and
  \cite[Sec.\,4]{MR3050164} taking into account the properties of the
  operator of multiplication in a de Branges space \cite{MR0229011}.
\end{remark}

For any regular, symmetric operator $A$ with $n_{+}(A)=n_{-}(A)=1$,
there is  $\mu$ in $\cH$ such that
\begin{equation*}
\cH = \ran(A-zI)\dot{+}\Span\{\mu\}
\end{equation*}
for all $z\in\complex\setminus S_{\mu}$, where
$\card S_{\mu}\le\card\nats$ (see \cite[Sec.\,2.2]{MR1466698},
\cite[Sec.\,2]{MR2607264}). The set $S_{\mu}$ turns out to be at most
countable since it is the zero set of the analytic function
$\inner{\xi_{A}(\,\cc{\,\boldsymbol{\cdot}\,}\,)}{\mu}$, which does
not vanish identically due to the fact that $\mu\ne 0$ and
$\{\xi_{A}(z)\}_{z\in\complex\setminus\reals}$ is a total set in $\cH$
(\cf\ \eqref{eq:completely-nonselfadjoint}). The vector $\mu$ is said
to be a gauge of $A$.

The gauge $\mu$ can be chosen in such a way so that the exceptional
set $S_{\mu}$ lies entirely on the real line
\cite[Lem.\,2.1]{MR2607264} or completely outside the real line
\cite[Thm.\,2.2]{MR2607264}. This last assertion was first stated
without proof in \cite[Thm.\,8]{MR0011533}.

\begin{definition}
  \label{def:entire-operator}
  A regular, symmetric operator $A$ such that $n_{+}(A)=n_{-}(A)=1$ is
  said to be \emph{entire} if there exists a gauge $\mu$ so that
  $S_{\mu}=\emptyset$. In this case $\mu$ is an \emph{entire gauge} of
  $A$.
\end{definition}

A straightforward consequence of this definition is that if $A$ is an
entire operator and $\mu$ its entire gauge, then the entire function
  \begin{equation}
    \label{eq:t-function}
    t(\boldsymbol{\cdot}):=
  \inner{\xi_{A}(\,\cc{\,\boldsymbol{\cdot}\,}\,)}{\mu}
  \end{equation}
  is a zero
  free function. Another direct deduction is the following assertion.
  \begin{lemma}
    \label{lem:real-gauge}
    The function $t$ given in \eqref{eq:t-function} is real if and
    only if
    $\mathcal{I}\mu=\mu$.
  \end{lemma}
  \begin{proof}
    One has
    \begin{equation}
    \label{eq:splited-equalities}
  \begin{split}
    \inner{\mu}{\xi_{A}(\cc{z})}&=\inner{\xi_{A}(z)}{\mu}\\
                            &=\inner{\mathcal{I}\xi_{A}(\cc{z})}{\mu}\\
                           &=\inner{\mathcal{I}\mu}{\xi_{A}(\cc{z})}\,,
  \end{split}
\end{equation}
where the first equality is actually
  \eqref{eq:real-function}, the second one follows from
  Proposition~\ref{prop:function-xi}(c), and in the third one,
  \eqref{eq:inversion-inner-conjugation} and the involutive property
  of the conjugation are used. Thus, on the basis of
  \eqref{eq:completely-nonselfadjoint}, the assertion follows from
  \eqref{eq:splited-equalities}.
  \end{proof}
  
  It is established in \cite[Ch.\,2 Sec.\,4.1]{MR1466698} that, for
  any entire operator $A$, the vector-valued function $\xi_{A}$ and
  the gauge $\mu$ can be chosen so that the scalar function $t$ given
  in \eqref{eq:t-function} is a real constant (see also
  \cite[Sec.\,2]{MR0012177}). Below, it will be shown that the
  ``natural'' choice of the gauge $\mu$ and the function $\xi$ for an
  entire operator is the one for which $t$ is a real constant.
  However, having done this choice, we are also interested in the
  behavior of the zero-free entire function
  $\inner{\xi_{A}(\cc{z})}{\tilde\mu}$ where the entire gauge $\mu$
  has been substituted by another entire gauge $\tilde{\mu}$. To this
  end, let us recall two notions related to the theory of growth of
  entire functions.

  A function of at most exponential type is a function of at most
  order one and normal type \cite[Ch.\,1 Sec.\,20]{MR589888}. The
  dependence of the growth of a function $f$ of exponential type on the
  direction in which the independent variable tends to infinity is
  given by the function
  \begin{equation*}
    h_{f}(\theta):=\limsup_{r\to\infty}\frac{\log\abs{f(re^{i\theta})}}{r}
    \qquad \theta\in [0,2\pi)
  \end{equation*}
  which is the so-called \emph{indicator function} of the function of
  exponential type $f$ (see \cite[Ch.\,1 Sec.\,15]{MR589888} and
  \cite[Ch.\,II.9 Sec.\,45]{MR0444912}).

  In \cite[Sec.\,8]{MR0048704} (see also \cite[Ch.\,2
  Sec.\,5]{MR1466698}), the following assertion is established.

 \begin{proposition}
   \label{prop:krein-growth-functions}
   Let $A$ be an entire operator and pick the corresponding function
   $\xi_{A}$ and gauge $\mu$ so that the function $t$ given in
   \eqref{eq:t-function} is a real constant. Then, for any
   $\phi\in \mathcal{H}$, the function
   $f(\boldsymbol{\cdot}):=\inner{\xi_{A}(\,\cc{\,\boldsymbol{\cdot}\,}\,)}{\phi}$
   is at most of exponential type and its indicator function obeys
   \begin{equation}
    \label{eq:indicator}
    h_{f}(\theta)=
    \begin{cases}
      h_{f}(\frac{\pi}{2})\sin\theta & \text{if}\quad 0\le\theta\le\pi\\
      -h_{f}(-\frac{\pi}{2})\sin\theta & \text{if}\quad \pi<\theta\le 2\pi\,.
    \end{cases}
  \end{equation}
 \end{proposition}
 The proof of the first part of
 Proposition~\ref{prop:krein-growth-functions} is found in the
 paragraph preceding \cite[Lem.\,8.1]{MR0048704} (see also \cite[Eq.\,
 5.1]{MR1466698}), where implicitly it is used that
 \eqref{eq:t-function} is a constant. As regards the second part see
 the proof of \cite[Lem.\,8.1]{MR0048704} or the proof of \cite[Ch.\,2
 Lem.\,5.1]{MR1466698}).

The following assertion exhibits a property of entire operators which
is crucial for this section. It is related to
\cite[Thm.\,1]{MR0012177} whose proof can be found in
\cite[Prop.\,4.6]{MR2607264}.

\begin{proposition}
  \label{prop:uniqueness-entire-gauge}
  Let $A$ be an entire operator. If there are two entire gauges of $A$
  and two functions $\xi_{A}^{(1)}$, $\xi_{A}^{(2)}$ satisfying
  (a)--(c) of Proposition~\ref{prop:function-xi} so that
  $\inner{\xi_{A}^{(1)}(\,\cc{\,\boldsymbol{\cdot}\,}\,)}{\mu_{1}}$
  and
  $\inner{\xi_{A}^{(2)}(\,\cc{\,\boldsymbol{\cdot}\,}\,)}{\mu_{2}}$
  are real constants, then there is a real constant $C$ such that
  $\mu_{1}=C\mu_{2}$.
\end{proposition}
\begin{proof}
  Lemma~\ref{lem:real-gauge} yields that $\mathcal{I}\mu_{1}=\mu_{1}$
  and $\mathcal{I}\mu_{2}=\mu_{2}$. Thus, the zero-free function
  $f(\boldsymbol{\cdot}):=
  \inner{\xi_{A}^{(1)}(\,\cc{\,\boldsymbol{\cdot}\,}\,)}{\mu_{2}}$ is
  real and has the form $\exp(g(\boldsymbol{\cdot}))$, where $g$ is a
  real entire function. Furthermore, since
  $\inner{\xi_{A}^{(1)}(\,\cc{\,\boldsymbol{\cdot}\,}\,)}{\mu_{1}}$ is
  a real constant, the first part of
  Proposition~\ref{prop:krein-growth-functions} implies that the
  function $f$ is a function of at most exponential type and therefore
  the function $g$ is a polynomial of the first degree, whence
  $\inner{\xi_{A}^{(1)}(\cc{z})}{\mu_{2}}=C\exp((a+\I b)z)$ for all
  $z\in\complex$.

  On the one hand, it follows from the second part of
  Proposition~\ref{prop:krein-growth-functions} that the indicator
  function of $f$ is \eqref{eq:indicator}; on the other hand, the
  indicator function of $\exp((a+\I b)\,\boldsymbol{\cdot}\,)$ has the
  form:
  \begin{equation}
    \label{eq:trigonometric-indicator}
    h(\theta)=a\cos\theta - b\sin\theta\,.
  \end{equation}
  Comparing \eqref{eq:indicator} with
  \eqref{eq:trigonometric-indicator}, one arrives at the conclusion
  that $a=0$. Finally, it follows from the reality of the function
  $\inner{\xi_{A}^{(1)}(\,\cc{\,\boldsymbol{\cdot}\,}\,)}{\mu_{2}}$
  that $b=0$ and $C$ is real.
\end{proof}

Any nonselfadjoint Jacobi operator is regular. This is a classical
result of the moment problem (or Jacobi operator) theory. It is shown
by establishing that the spectra of its selfadjoint extensions do not
intersect (see the proof of \cite[Thm.\,4.2.4]{MR0184042} and
\cite[Thm. 5]{MR1627806}). Recall that the spectral kernel of an
operator is contained in the spectral kernel of its extension, thus if
a point is in the spectral kernel of a symmetric operator, then this
point is in the spectrum on any of its selfadjoint extensions.

\begin{proposition}
  \label{prop:jacobi-entire}
  Any nonselfadjoint Jacobi operator is an entire operator.
\end{proposition}
\begin{proof}
  It has been established that any nonselfadjoint operator $J$ having
  the matrix representation \eqref{eq:jacobi-matrix} with respect to
  the orthonormal basis $\{\delta_{k}\}_{k=1}^{\infty}$ has deficiency
  indices $n_{+}(J)=n_{-}(J)=1$ and it is regular. As mention in
  Remark~\ref{rem:on-kernel-adjoint-at-zeta}, the vector-valued
  function $\pi$ satisfies (b) of
  Proposition~\ref{prop:function-xi}. Also, it follows from the fact
  that the zeros of polynomials of the first kind interlace
  \cite[Thm.\,1.2.2]{MR0184042} that $\pi$ complies with (a) of
  Proposition~\ref{prop:function-xi}. The property (c) is a
  consequence of the reality of the polynomials' coefficients.
   For finishing the proof it only
  remains to note that
  $\inner{\pi(\,\cc{\,\boldsymbol{\cdot}\,}\,)}{\delta_{1}}\equiv 1$.
\end{proof}

\begin{proposition}
  \label{prop:only-one-basis}
  There is only one basis of representation (modulo
  reflection\footnote{Reflection means that every element of the
    orthonormal basis is multiplied by $-1$.}) with respect to which
  any nonselfadjoint Jacobi operator has a Jacobi matrix as its matrix
  representation. The Jacobi matrix representing a nonselfadjoint
  Jacobi operator is unique.
\end{proposition}
\begin{proof}
  Let $J$ be nonselfadjoint and have the matrix representation
  \eqref{eq:jacobi-matrix} with respect to the orthonormal basis
  $\{\delta_{k}\}_{k=1}^{\infty}$. As has been shown, $\delta_{1}$
  determines all the elements of the orthonormal basis and,
  consequently, the entries of the Jacobi matrix. Indeed, the vectors
  $\delta_{2},\delta_{3},\dots$ are obtained by applying the
  Gram-Schmidt procedure to the sequence
  $\{J^{k-1}\delta_{1}\}_{k=1}^{\infty}$ (see the proof of
  Proposition~\ref{prop:basis-representation}). Likewise, as asserted
  in Remark~\ref{rem:relation-moments-matrices}, the entries
  of the matrix can be obtained from the moments
  $\inner{\delta_{1}}{J^{k-1}\delta_{1}}$, $k\in\nats$. Now, suppose that
  for $J$ there is another orthonormal basis
  $\{\tilde\delta_{k}\}_{k=1}^{\infty}$ with respect to which $J$ has
  a Jacobi matrix representation. As shown in the proof of
  Proposition~\ref{prop:jacobi-entire}, the vector $\tilde\delta_{1}$
  is an entire gauge of $J$ satisfying the hypothesis of
  Proposition~\ref{prop:uniqueness-entire-gauge}. Therefore
  $\tilde\delta_{1}=C\delta_{1}$, where $C\in\reals$. Since
  $\norm{\delta_{1}}=\|\tilde\delta_{1}\|=1$, one concludes that $C$
  is either $1$ or $-1$.
\end{proof}
\subsection*{Acknowledgments}
D.H.B. is supported with a postdoctoral fellowship by DGAPA-UNAM at
IIMAS-UNAM. L.O.S. has been partially supported by CONACyT Ciencia de
Frontera 2019 \textnumero 304005. The authors thank Professor R.
Szwarz for his comments and interest in this work.

\def\cprime{$'$} \def\lfhook#1{\setbox0=\hbox{#1}{\ooalign{\hidewidth
  \lower1.5ex\hbox{'}\hidewidth\crcr\unhbox0}}} \def\cprime{$'$}

\end{document}